\documentclass{mcma_v02_mod}
\usepackage{dsfont}
\usepackage{setspace}
\usepackage{algorithmic}
\usepackage{multirow}
\usepackage{hyperref}
\title{Exact Simulation of Bessel Diffusions}
\headlinetitle{Exact Simulation of Bessel Diffusions}

\authorone{Roman N. Makarov} %\footnote{The corresponding author.}
\addressone{Department of Mathematics, Wilfrid Laurier University\\
75 University Avenue West, Waterloo, Ontario}
\countryone{Canada}
\emailone{rmakarov@wlu.ca}

\authortwo{Devin Glew}
\addresstwo{University of Waterloo\\ 200 University Avenue West,
Waterloo, Ontario}
\countrytwo{Canada}
\emailtwo{dtglew@math.uwaterloo.ca}

\def\E{\mathbb{E}}
\def\P{\mathbb{P}}
\def\R{\mathbb{R}}
\def\C{\mathbb{C}}
\def\s{\mathfrak{s}}
\def\m{\mathfrak{m}}
\def\I{\mathcal{I}}
\def\M{\mathcal{M}}
\def\U{\mathcal{U}}

\def\G{\mathcal{G}}
\def\F{\mathsf{F}}
\def\X{\mathsf{X}}

\abstract{We consider the exact path sampling of the squared Bessel process and some other continuous-time Markov processes, such as the CIR model, constant elasticity of variance diffusion model, and hypergeometric diffusions, which can all be obtained from a squared Bessel process by using a change of variable, time and scale transformation, and/or change of measure. All these diffusions are broadly used in mathematical finance for modelling asset prices, market indices, and interest rates.
We show how the probability distributions of a~squared Bessel bridge and a~squared Bessel process with or without absorption at zero are reduced to randomized gamma distributions. Moreover, for absorbing stochastic processes, we develop a new bridge sampling technique based on conditioning on the first hitting time at zero. Such an approach allows us to simplify simulation schemes. New methods are illustrated with pricing path-dependent options.}



\acknowledgments{The authors acknowledge the support of the Natural Sciences and
Engineering Research Council of Canada (NSERC) for a discovery
research grant and an undergraduate student research award.}

\begin{document}

\maketitle

\section{Introduction} \label{sect1}
In this paper we study the exact path simulation of solvable continuous-time stochastic processes with transition probability density functions being obtainable in analytically closed-form. Despite the popularity of various approximation schemes for stochastic differential equations (SDEs), the \emph{precise} path sampling of continuous-time Markov processes has certain advantages. Sampling from the exact probability distribution allows us to avoid introducing a bias and also to integrate along a path over an arbitrarily long time horizon. %Approximation schemes for SDEs such as the Euler method have certain restrictions on the SDE's coefficients and sometimes can not be directly applied to many interesting stochastic models.

Our main motivation is the Monte Carlo pricing of path-dependent financial derivatives. The no-arbitrage price of a European-style option takes the form of a multidimensional integral along a path of an underlying asset price process. The usual procedure to the evaluation of such an integral is to employ the Monte Carlo method.
%Notice that a change of measure approach allows us to simplify the path simulation procedure (e.g., see \cite{CM08}).
Pricing of an American-style option reduces to solving a dynamic-programming problem. Therefore, to apply the Monte Carlo method we have to sample paths from the exact distribution of the asset price process (e.g., see~\cite{GL04}).

More specifically, we study continuous-time Markov processes that arise from a squared Bessel (SQB) diffusion such as the squared radial Ornstein-Uhlenbeck process (known also as the Cox-Ross-Ingersoll model), the constant-elasticity of diffusion model (with a power volatility function), and so-called hypergeometric diffusions obtained from the squared Bessel process by means of a special combination of a change of measure and changes of variables (see \cite{CM07,CM08,CM09}). All these stochastic processes are broadly used in mathematical finance. Although for these models many fundamental quantities such as probability distributions of the first-hitting time at a barrier, maximum and minimum values, and pricing formulas for barrier and lookback options can be obtained in closed-form, the Monte-Carlo method remains an important tool for the verification of analytical formulas and also for pricing Asian and American derivatives.

As is shown in \cite{Yuan}, the transition probability distributions of a squared Bessel process (without absorption at zero) and a squared Bessel bridge relate to the so-called randomized gamma distributions, which are mixture gamma distributions with a random rate parameter.
%Therefore, the path simulation of a squared Bessel process reduces to sampling from the gamma, Poisson, and Bessel probability distributions. Recall that the Bessel distribution is a power series distribution generated by the modified Bessel function of the first kind.
The simulation of an~SQB process with absorption at the origin is less studied in the literature. As is shown in \cite{CM07}, the normalized transition density function of the SQB process is a gamma density which is randomized by a discrete probability distribution generated by a power series expansion of the lower incomplete gamma function. Therefore, to sample an increment of the random process we first simulate the absorption event and then sample from the normalized density function in case of surviving. Since we are able to derive the first-hitting time distribution of the SQB process with absorption at zero, it is possible to implement a completely different approach. First, we sample the first-hitting time, $\tau_0$, at the origin. After that, we sample the Bessel bridge with its value at time $\tau_0$ tied at zero. We show that the simplest realization of such an approach allows us to sample a path of the SQB process by only employing the gamma and Poisson probability distributions.

The paper is organized as follows. Section~2 gives some basis results about the squared Bessel process and the squared Bessel bridge. Section~3 provides different sampling algorithms. In Section~4, we introduce other diffusion processes arising from the SQB process and provide simulation algorithms for them. Section~5 contains some numerical results.

\section{The Squared Bessel Process and Bessel Bridge}\label{sect2}

\subsection{The Squared Bessel Process} \label{subsect2.1}
Let us consider a $\lambda_0$-dimensional squared Bessel (SQB) process $(X_t)_{t\geq 0}$ obeying the stochastic differential equation (SDE)
\begin{equation}\label{Bessel}
  dX_t=\lambda_0 dt+\nu\sqrt{X_t}dW_t,\;X_t\in\I=(0,\infty),
\end{equation} with constant parameters $\lambda_0$ and $\nu>0$.  The scale and speed
densities are respectively $\s(x)=x^{-\mu-1}$
and $\m(x)=\frac{2}{\nu^2}x^{\mu},$ where
$\mu\equiv\frac{2\lambda_0}{\nu^2}-1$ is called the index of the process. The left-hand boundary $l=0$ is entrance if $\mu\geq0$, regular if $-1<\mu<0$, or exit if $\mu\leq -1$. The right-hand boundary $r=\infty$ is natural. For the regular diffusion on $\I$ the transition probability density function (PDF) is given by
\begin{equation}
p(t;x,y) \equiv \textstyle\frac{\P(X_t\in dy| X_0=x)}{dy} = \displaystyle\left({y\over
x}\right)^{\frac{\mu}{2}}
    \,{e^{-2(x + y)/\nu^2t} \over \nu^2t/2}
     I_{\tilde{\mu}}\left({4\sqrt{xy}\over \nu^2t}\right).
 \label{PrkernelSQB}
\end{equation}
where $\tilde{\mu}=\mu$ if $l=0$ is entrance or a regular reflecting boundary, and $\tilde{\mu}=|\mu|$ if $l=0$ is exit or a regular killing boundary.

For simplicity of presentation, we assume here that $\nu=2$. A simple scale transformation $X_t^{(\nu'_0,\lambda'_0)} = \left(\frac{\nu'_0}{\nu''_0}\right)^2 X_t^{(\nu''_0,\lambda''_0)},$ $\lambda'_0=\lambda''_0\left(\frac{\nu'_0}{\nu''_0}\right)^2$, allows us to modify $\nu$ without changing  $\mu$ (i.e. $\mu'=\mu''$).

\subsection{The First Hitting Time Distribution} \label{subsect2.2}
In the case when $l=0$ is an absorbing boundary ($\mu<0$, $\tilde{\mu}=|\mu|$), the density in (\ref{PrkernelSQB}) does not satisfy probability conservation on $\I$.
The first hitting time (FHT), $\tau_0$, at zero for the SQB process $(X_t)$ starting at $x_0$ is defined by $\tau_0 = \inf\{t\;:\; X_t=0\mid X_0=x_0\}$. The PDF $q(x_0;\tau)$ for the FHT distribution is given by
\begin{equation} \label{FHTPDF} q(x_0;\tau) = -\frac{\partial}{\partial \tau} \int_0^\infty p(\tau;x_0,x)dx. \end{equation}
By using that the transition PDF $p$ satisfies Kolmogorov equations, we simplify the expression in (\ref{FHTPDF}) to obtain
\begin{equation} \label{FHTPDFlim}  q(x_0;\tau) = \frac{1}{\s(x)} \frac{\partial}{\partial  x}\left(  \frac{p(\tau;x_0,x)}{\m(x)}\right)\bigg|_{x=\infty}^{x=0+}. \end{equation}
As a result, we derive a closed-form expression for the FHT PDF:
\begin{equation}\label{SQBFHTPDF}
    q(x_0;\tau) = \displaystyle\frac{1}{\tau\Gamma(|\mu|)} \left(\frac{x_0}{2\tau}\right)^{|\mu|} \exp\left( -\frac{x_0}{2\tau}\right).
\end{equation}
A simple change of variable reduces the PDF in (\ref{SQBFHTPDF}) to that of the gamma distribution $\mathrm{G}(\alpha,\beta)$ with shape parameter $\alpha=|\mu|$ and rate parameter $\beta=1$. Therefore, the FHT, $\tau_0$, can be sampled by using the formula
$\tau_0 = \frac{x_0}{2Y},$ where  $Y\sim \mathrm{G}(|\mu|,1).$

\subsection{The Squared Bessel Bridge} \label{subsect2.3}
Let $0\leq t_1<t<t_2$. Consider a stochastic bridge generated by a continuous-time Markov process $(X_t)_{t\geq 0}\in\I$ with $X_{t_1}$ and $X_{t_2}$ tied at $x_1$ and $x_2$, respectively. The bridge PDF $b$ defined by $b(t_1,t_2,t;x_1,x_2,x)dx=\P\{X_t\in dx| x_{t_1}=x_1,X_{t_2}=x_2\}$ can be expressed in terms of the transition PDF $p$ of $(X_t)$ as follows:
\begin{equation} \label{BRGPDF}
   b(t_1,t_2,t;x_1,x_2,x) = \frac{p(t-t_1;x_1,x)p(t_2-t;x,x_2)}{p(t_2-t_1;x_1,x_2)}.
\end{equation}
Clearly, the bridge PDF $b$ in (\ref{BRGPDF}) integrates to unity thanks to the Chapman-Kolmogorov equation
$ p(t_2-t_1;x_1,x_2) = \int_\I p(t-t_1;x_1,x)p(t_2-t;x,x_2) dx$.
Notice that for the bridge density of a Gaussian process may also be derived in closed form by using a conditional multivariate normal distribution.

The PDF of the squared Bessel bridge $(X_t)_{0\leq t\leq T}$ conditional on $X_0=x$ and $X_T=z$ is given by
\begin{equation}
 b(0,T,t;x,z,y) =\frac{T}{2t(T-t)}
 e^{-\displaystyle\frac{\bar{x}+\bar{y}}{2t}-\frac{\bar{z}t}{2}}\frac{I_{\tilde{\mu}}(\sqrt{\bar{x}\bar{y}}/t)I_{\tilde{\mu}}(\sqrt{\bar{y}\bar{z}}/(T-t))}{I_{\tilde{\mu}}(\sqrt{\bar{x}\bar{z}}/T)},
 \label{BRGPDFfull}
\end{equation}
where $\bar{x} \equiv \frac{x\,(T-t)}{T}$, $\bar{y} \equiv \frac{y\,T}{T-t}$, and $\bar{z} \equiv \frac{z}{T(T-t)}$, $0<t<T.$

Suppose that $X_t$ is sampled conditionally on the FHT, $T=\tau_0$. If $t\geq\tau_0$, then set $X_t=0$. Otherwise, if $t<\tau_0$, we use the Bessel bridge with $X_0$ and $X_{T=\tau_0}$ tied at $x$ and $z=0$, respectively. In the limiting case as $z\to 0+$ in (\ref{BRGPDFfull}), we obtain
\begin{equation}
  b(0,T,t;x,0,y) = \frac{T}{2t(T-t)} \left(\frac{\bar{y}}{\bar{x}}\right)^{\tilde{\mu}/2}\,\exp\left(-\frac{\bar{x} + \bar{y}}{2t}\right) I_{\tilde{\mu}}\left(\frac{\sqrt{\bar{x}\bar{y}}}{t}\right).  \label{BRGPDF0}
\end{equation}
Notice that the PDF in (\ref{BRGPDF0}) has the same form as that in (\ref{PrkernelSQB}).

\section{Simulation Algorithms} \label{sect3}
In this section we present several algorithms for the precise path generation of the SQB process $(X_t)$. That is, for every time partition $0=t_0<t_1<\cdots<t_N$, $N\geq 1$, we sample a path-skeleton $\mathbf{X}\equiv(X_0,X_1,\ldots,X_N)$, $X_n\equiv X_{t_n}$, from the exact multivariate probability distribution. The algorithms proposed below are all based on sampling from a randomized gamma distribution of the form $\mathrm{G}(\alpha+Y,\beta)$, where $\alpha+Y>0$ and $\beta>0$ are scale and rate parameters, respectively, and $Y$ is a nonnegative integer-valued random variable.
As is mentioned above, we assume that $\nu=2$, so all algorithms presented below deal with this case. In the general situation when $\nu\neq 2$, we proceed as follows. For given $\lambda_0,$ $\nu$, $X_0$, sample a path of the SQB process  with $\mu=2\lambda_0/\nu^2-1$ that starts at $\left(\frac{2}{\nu}\right)^2X_0$ by using one of algorithms in Figures~\ref{AlgSQB1}--\ref{AlgSQB4}. After that, rescale the path obtained by multiplying its values by $\left(\frac{\nu}{2}\right)^2$.

\subsection{Randomized Gamma Distributions} \label{subsect3.1}
Suppose that a discrete random variable $Y$ has discrete probabilities $\P\{Y=n\}=p_n,$ $n=0,1,2,\ldots.$ The PDF $f$ of the mixture probability distribution $\mathrm{G}(\alpha+Y,\beta)$ admits the form of a series expansion:
$f(x) = \sum_{n=0}^\infty p_n \frac{\beta^{\alpha+n}}{\Gamma(\alpha+n)}x^{\alpha+n-1} e^{-\beta x}.$

Let us consider three choices for the randomizer $Y$ of the gamma distribution $\mathrm{G}(\alpha+Y,\beta)$. The resulting distributions are called the randomized gamma distribution of the first, second, and third types, respectively.

Let $Y_1\sim \mathrm{P}(\lambda)$ be a Poisson random variable with mean $\lambda>0$.
%and probabilities
%\begin{equation} \label{ProbPoisson} \P\{Y_1=n\}=\frac{\lambda^{-n}e^{-\lambda}}{n!},\quad
%n=0,1,2,\ldots.
%\end{equation}
The randomized gamma distribution of the {\it first type} is $\mathrm{G}(Y_1+\theta+1,\beta)$, $\theta>-1,$ $\beta>0,$ with the PDF
\begin{equation}
f_1(y)=\beta\left(\frac{\beta}{\lambda}\right)^{\theta/2} y^{\theta/2}e^{-\lambda-\beta y}
   I_\theta(\sqrt{4\beta\lambda y}),
   \quad y>0.
\label{RandG1}
\end{equation}

A discrete random variable $Y_2$ is said to have a Bessel
probability distribution $\mathrm{Bes}(\theta,b)$ with parameters $\theta>-1$ and $b>0$ if
\begin{equation} \label{ProbBessel} \P\{Y_2=n\}=\frac{(b/2)^{2n+\theta}}{I_\theta(b)\;n!\;\Gamma(n+\theta+1)},\quad
n=0,1,2,\ldots.
\end{equation}
This distribution is related to many other distributions, where the Bessel function $I$ is involved in the density, including the squared Bessel bridge distribution (see
\cite{Yuan} for details).
The randomized gamma distribution of the {\it second type} is a
mixture distribution $\mathrm{G}(Y_1+2Y_2+\theta+1,\beta),$ $\beta>0,$
$\theta>-1,$ where $Y_1\sim
\mathrm{P}((a+b)/(4\beta))$ and $Y_2\sim \mathrm{Bes}(\theta,\sqrt{ab}/(2\beta))$ are independent Poisson and Bessel variates,
respectively. For any positive numbers $\beta$, $a,$ $b,$ and $\theta>-1$, the PDF is
\begin{equation}
   f_2(y)=\frac{\beta}{I_\theta(\sqrt{ab}/(2\beta))}
   e^{-(a+b)/4\beta-\beta y}I_\theta(\sqrt{ay})I_\theta(\sqrt{by}),
   \quad y>0.
\label{RandG2}
\end{equation}

A discrete random variate $Y_3$ is said to follow an {\it incomplete Gamma} probability distribution, which we simply denote by $\mathrm{I}\Gamma(\theta,\lambda)$ with
parameters $\lambda>0$ and $\theta>0$, if
\begin{equation}
\P\{Y_3=n\}=e^{-\lambda}\frac{\lambda^{n+\theta}}{
   \Gamma(n+\theta+1)} \frac{\Gamma\left(\theta\right)}{\gamma\left(\theta,\lambda\right)},\quad n=0,1,2,\ldots.
\label{ProbIG}
\end{equation}
Notice that if $\theta=0,1,2,\ldots$, then the distribution of $Y_3$ is a truncated and shifted Poisson distribution thanks to the property \[\frac{\gamma\left(m,a\right)}{\Gamma\left(m\right)} = 1- \left( 1+x+ \ldots+\frac{x^{m-1}}{(m-1)!}\right)e^{-x},\; m=0,1,2,\ldots.\]

We call a mixture Gamma distribution $\mathrm{G}(Y_3+1,\beta)$, $Y_3\sim\mathrm{I}\Gamma(\theta,\lambda)$, the randomized gamma distribution of the \emph{third type}. The PDF is
\begin{equation}
f_3(y)=\beta\frac{\Gamma\left(\theta\right)}{\gamma\left(\theta,\lambda\right)}\left(\frac{\beta}{\lambda}\right)^{-\theta/2} y^{-\theta/2}e^{-\lambda-\beta y}
   I_\theta(\sqrt{4\beta\lambda y}),
   \quad y>0.
\label{RandG3}
\end{equation}

\subsection{Simulation of Processes without Absorption} \label{subsect3.2}
\begin{figure}
\centering
 \framebox[0.65\linewidth]{\begin{minipage}{0.65\linewidth}
    \begin{algorithmic}
    \STATE \textbf{input} $X_0>0,$ $0=t_0<t_1<\cdots<t_N,$ $\mu>-1$
    \FOR {$n$ from 1 to $N$}
    \STATE  $Y_n\sim \mathrm{P}\left(\displaystyle\frac{X_{n-1}}{2
 (t_n-t_{n-1})}\right)$
    \STATE  $X_n\sim \mathrm{G}\left(Y_n+\mu+1,
                          \displaystyle\frac{1}{2(t_n-t_{n-1})} \right)$
    \ENDFOR
    \RETURN $(X_0, X_1,\ldots, X_N)$
    \end{algorithmic}
 \end{minipage}}

  \caption{ \label{AlgSQB1} \small The sequential sampling method for modeling an SQB process without absorption.}
\end{figure}

The randomized distribution of the first type is closely connected with the transition distribution of a squared
Bessel process $(X_t)$ without absorption (i.e. $\mu\geq 0$, or $\mu\in(-1,0)$ and $x=0$ is a reflecting boundary). The conditional distribution of $X_t$, $t>0$, given $X_0=x_0>0$, is then a randomized gamma distribution of
the first type. The transition PDF in (\ref{PrkernelSQB}) with $\nu=2$ has the form of the PDF $f_1$ in (\ref{RandG1}) with $\theta=\mu$, $\beta=1/2t$, and $\lambda=x_0/2t$. Therefore, we have the following sampling scheme:
\begin{equation} \label{simbes}
X_t\sim \mathrm{G}(\mu+Y+1,1/2t),\mbox{ where }\, Y\sim \mathrm{P}(x_0/2t),\;t>0.
\end{equation}
The sampling algorithm is presented in Figure~\ref{AlgSQB1}.

A path of the standard squared Bessel bridge can be generated using the second type randomized gamma distribution. The bridge PDF in
(\ref{BRGPDFfull}) reduces to that in (\ref{RandG2}) by
setting $a\equiv x/t^2$, $b\equiv z/(T -t)^2$,
$\beta\equiv \frac{T}{2t(T-t)}$, and $\theta=\mu$. Then, $X_t$ conditional on $X_0=x$ and $X_T=z$, $0<t<T$, can be obtained by
generating two independent random variables $Y\sim \mathrm{P}\left( \frac{1}{2T} \left[ \frac{T-t}{t}x+\frac{t}{T-t}z \right] \right)$ and $Z\sim \mathrm{Bes}\left(\mu,\frac{\sqrt{xz}}{T}\right)$, and then $X_t\sim \mathrm{G}\left(Y+2Z+\mu+1,\frac{T}{2t(T-t)}\right)$.

\subsection{Sequential Simulation of Processes with Absorption} \label{subsect3.3}

\begin{figure}
\centering
 \framebox[0.65\linewidth]{\begin{minipage}{0.65\linewidth}
    \begin{algorithmic}
    \STATE \textbf{input} $X_0>0,$ $0=t_0<t_1<\cdots<t_N,$ $\mu<0$
    \STATE $\tilde\tau_0 \gets \infty$
    \FOR {$n$ from 1 to $N$}
      \IF{$\tilde\tau_0=\infty$}
        \STATE $p_a\gets\Gamma\left(|\mu|,\displaystyle\frac{X_{n-1}}{2(t_n-t_{n-1})}\right)/\Gamma(|\mu|)$
        \STATE  $U_n\sim \mathrm{U}(0,1)$
        \STATE \textbf{if} $U_n<p_a$ \textbf{then} $\tilde\tau_0\gets t_n$
      \ENDIF
      \IF{$t_n<\tilde\tau_0$ }
        \STATE  $Y_n\sim \mathrm{I\Gamma}\left(|\mu|,\displaystyle\frac{X_{n-1}}{2(t_n-t_{n-1})}\right)$
        \STATE  $X_n\sim \mathrm{G}\left(Y_n+1,\displaystyle\frac{1}{2(t_n-t_{n-1})}\right)$
      \ELSE
        \STATE $X_n\gets0$
      \ENDIF
    \ENDFOR
    \RETURN $(X_0, X_1,\ldots, X_N)$ and $\tilde\tau_0$
    \end{algorithmic}
 \end{minipage}}
  \caption{ \label{AlgSQB2} \small The sequential sampling method for an~SQB process with absorption at the origin.}
\end{figure}

Assume that a stochastic process $(X_t)_{t\ge 0}\in R_+$ admits absorption at the origin. For example, for an~SQB process we have that $\mu<0$ and $x=0$ is a killing boundary or exit. Clearly, the transition PDF $p$ given by (\ref{PrkernelSQB}) with $\tilde{\mu}=|\mu|$, $\mu<0$,  does not integrate to one. Let us define the probability $P_s$ of surviving before time $t$  and the probability $P_a$ of absorption before time $t$ for the process $(X_t)$ started at $X_0=x$:
\[ P_s(x;t) = \int_0^\infty p(t;x,y)dy>0 \mbox{ and } \,P_a(x;t) = 1-P_s(x;t)>0. \]
Observe that the actual transition probability distribution is then a mixture of continuous and discrete probability distributions with the following generalized PDF:
\[ p(X_0\to X_t) = P_s(X_0;t) \cdot \left( \frac{p(t;X_0,X_t)}{P_s(X_0;t)} \right) + P_a(X_0;t) \cdot \delta(X_t), \]
where $\delta$ denotes a delta function.

\begin{figure}
\centering
 \framebox[0.75\linewidth]{\begin{minipage}{0.75\linewidth}
    \begin{algorithmic}
    \STATE \textbf{input} $X_0>0,$ $0=t_0<t_1<\cdots<t_N,$ $\mu<0$
    \STATE  $Y\sim \mathrm{G}(|\mu|,1)$, \quad $\tau_0 \gets \displaystyle\frac{X_0}{2Y}$
    \FOR {$n$ from 1 to $N$}
      \IF{$t_n<\tau_0$ }
        \STATE  $Y_n\sim \mathrm{P}\left(\displaystyle\frac{X_{n-1}(\tau_0-t_n)}{2 (\tau_0-t_{n-1})(t_n-t_{n-1})}\right)$
        \STATE  $X_n\sim \mathrm{G}\left(Y_n+|\mu|+1, \displaystyle\frac{\tau_0-t_{n-1}}{(\tau_0-t_n)(t_n-t_{n-1})} \right)$
      \ELSE
        \STATE $X_n\gets0$
      \ENDIF
    \ENDFOR
    \RETURN $(X_0, X_1,\ldots, X_N)$ and $\tau_0$
    \end{algorithmic}
 \end{minipage}}
  \caption{ \label{AlgSQB3} \small The sequential sampling method conditional on the FHT, $\tau_0$, for modeling an~SQB process with absorption at the origin.}
\end{figure}

By using (\ref{SQBFHTPDF}), we obtain the following probabilities of surviving and absorption of the SQB process before time $t$:
\[
  P_s(x;t) =  \P\{\tau_0>t\} = \frac{\gamma\left(|\mu|,\frac{x}{2t}\right)}{\Gamma(|\mu|)}\mbox{ and }
  P_a(x;t) = \P\{\tau_0\leq t\} = \frac{\Gamma\left(|\mu|,\frac{x}{2t}\right)}{\Gamma(|\mu|)},
\]
where $\gamma(a,x)$ and $\Gamma(a,x)$ are the lower and upper incomplete gamma functions, respectively.
The normalized transition PDF of the SQB process conditioned on the survival of the process before time $t$ is
\begin{equation} \label{normSQBPDF} \frac{p(t;x,y)}{P_s(x;t)} = \frac{\Gamma\left(|\mu|\right)}{\gamma\left(|\mu|,\frac{x}{2t}\right)}\left({x\over
x_0}\right)^{\frac{\mu}{2}}
    \,{e^{-(x + x_0)/2t} \over 2t}
     I_{|\mu|}\left({\sqrt{xx_0}\over t}\right)\,.
\end{equation}
As is seen, the function in the right-hand side of (\ref{normSQBPDF}) reduces to the form of (\ref{RandG3}) with $\theta=|\mu|$, $\lambda=x/2t$, and $\beta=1/2t$.
Thus, the above normalized
transition PDF follows the randomized gamma distribution of the third kind $\mathrm{G}(Y+1,1/2t)$, where $Y\sim \mathrm{I}\Gamma(|\mu|,x/2t)$.
As a result, we obtain the sampling algorithm given in Figure~\ref{AlgSQB2}. The algorithm returns a sample path $\mathbf{X}$ and an approximation, $\tilde\tau_0\in\{t_1,\ldots,t_N,\infty\}$, of the FHT, $\tau_0$.

\subsection{Bridge Simulation of Processes with Absorption} \label{subsect3.4}
Consider again the~SQB process $(X_t)$ with absorption at the origin. Since the first hitting time PDF $q(x_0;\tau)$ is available, we may first sample the FHT, $\tau_0$, and then simulate a path of $(X_t)_{t \geq 0}$ conditional on $\tau_0$ by using the bridge distribution. As is seen from (\ref{BRGPDF0}), the PDF of $X_t$, $0<t<\tau_0$, conditional on $X_0=x$ and $X_{\tau_0}=0$ is reduced to the PDF $f_1$ in (\ref{RandG1}) of the randomized gamma distribution of the first type with $\theta=|\mu|$, $\lambda=\frac{x(\tau_0-t)}{2\tau_0 t}$, and $\beta=\frac{\tau_0}{2t(\tau_0-t)}$. As a result, we obtain a sequential sampling algorithm conditional on the FHT (see Figure~\ref{AlgSQB3}).

At last, in Figure~\ref{AlgSQB4}, we provide the full bridge sampling algorithm, where a path $\mathbf{X}=(X_0,X_1,\ldots,X_N)$, $N=2^k$, $k\ge 1$, is sampled at the time points in the following order of generation:
\[
   t_N, t_{N/2}, \underbrace{t_{N/4}, t_{3N/4}},
   \underbrace{t_{N/8}, t_{3N/8}, t_{5N/8}, t_{7N/8}},
   \ldots,
   \underbrace{t_2,t_6,\ldots, t_{N-2}}, \underbrace{t_{1}, t_{3},\ldots, t_{N-1}}.
\]
Here, we use that the bridge PDF in (\ref{BRGPDFfull}) with $\tilde\mu=|\mu|$ reduces to that in (\ref{RandG2}) by setting $a\equiv x/t^2$, $b\equiv z/(T -t)^2$,
$\beta\equiv \frac{T}{2t(T-t)}$, and $\theta=|\mu|$. Such a bridge sampling algorithm is very useful for the quasi-Monte Carlo pricing of path-dependent options.

\begin{figure}
\centering
 \framebox[0.9\linewidth]{\begin{minipage}{0.9\linewidth}
    \begin{algorithmic}
    \STATE \textbf{input} $X_0>0,$ $0=t_0<t_1<\cdots<t_N,$ $N=2^k$, $\mu<0$
    \STATE  $Y\sim \mathrm{G}(|\mu|,1)$, \quad $\tau_0 \gets \displaystyle\frac{X_0}{2Y}$
    \IF {$t_N<\tau_0$}
        \STATE  $Y_N\sim \mathrm{P}\left(\displaystyle\frac{X_0(\tau_0-t_N)}{2 \tau_0 t_N}\right)$, \quad  $X_N\sim \mathrm{G}\left(Y_n+|\mu|+1, \displaystyle\frac{\tau_0}{t_N(\tau_0-t_N)} \right)$
    \ELSE
      \STATE $X_N\gets 0$
    \ENDIF
    \FOR {$l$ from 1  to $k$}
      \FOR {$m$ from 1  to $2^{l-1}$}
        \STATE $n=(2m-1)2^{k-l}$
        \IF{$t_n\geq\tau_0$}
          \STATE $X_n\gets 0$
        \ELSE
          \STATE $n_1\gets {n-2^{k-l}}$, \quad $n_2\gets {n+2^{k-l}}$
          \IF{$t_{n_2}\geq\tau_0$}
            \STATE  $Y_n\sim \mathrm{P}\left(\displaystyle\frac{X_{n_1}(\tau_0-t_n)}{2 (\tau_0-t_{n_1})(t_n-t_{n_1})}\right)$
            \STATE  $X_n\sim \mathrm{G}\left(Y_n+|\mu|+1,
                          \displaystyle\frac{\tau_0-t_{n_1}}{(\tau_0-t_n)(t_n-t_{n_1})} \right)$
            \ELSE
            \STATE  $Y_n\sim \mathrm{P}\left(\displaystyle\frac{X_{n_1}(t_{n_2}-t_n)}{2 (t_{n_2}-t_{n_1})(t_n-t_{n_1})}+\frac{X_{n_2}(t_{n}-t_{n_1})}{2 (t_{n_2}-t_{n_1})(t_{n_2}-t_{n})}\right)$
            \STATE  $Z_n\sim \mathrm{Bes}\left(|\mu|,\displaystyle\frac{X_{n_1}(\tau_0-t_n)}{2 (\tau_0-t_{n_1})(t_n-t_{n_1})}\right)$            \STATE  $X_n\sim \mathrm{G}\left(Y_n+2Z_n+|\mu|+1,
                          \displaystyle\frac{t_{n_2}-t_{n_1}}{2(t_{n}-t_{n_1})(t_{n_2}-t_n)} \right)$
          \ENDIF
        \ENDIF
      \ENDFOR
    \ENDFOR
    \RETURN $(X_0, X_1,\ldots, X_N)$ and $\tau_0$
    \end{algorithmic}
 \end{minipage}}
  \caption{ \label{AlgSQB4} \small The full bridge sampling method conditional on the FHT, $\tau_0$, for modeling an~SQB process with absorption at the origin.}
\end{figure}

\section{Generating Paths of the CIR, CEV, and Hypergeometric Diffusions}\label{sect4}

\subsection{The CIR Process}\label{subsect4.1}
Consider the Cox-Ingerssol-Ross (CIR) diffusion process
$(Y_t)_{t\ge 0}\in\I=R_+$ solving the SDE
\begin{equation} \label{CIRmodel}
   dY_t=(\lambda_0-\lambda_1Y_t)dt+\nu\sqrt{Y_t}dW_t\,,
\end{equation}
where constant parameters $\lambda_0$, $\lambda_1$, and $\nu>0$. % (see \cite{CIRref}).
The respective scale and speed densities are $\s(x) =
x^{-\mu-1}e^{\kappa x}$  and $\m(x)=
\frac{2}{\nu^2}x^{\mu}e^{-\kappa x},$ where $\kappa \equiv
\frac{2\lambda_1}{\nu^2}.$
The boundary classification of the CIR process is equivalent that of
the SQB process. For the regular diffusion on $\I$, the transition PDF is
\begin{equation} \label{CIRden}
 p(t;x,y) = c_te^{\lambda_1 t}\left(\frac{ye^{\lambda_1 t}}{x}\right)^{\mu/2}
    \,e^{-c_t(ye^{\lambda_1t} + x)}
     I_{\tilde\mu}\left(2c_t\sqrt{ xye^{\lambda_1t} }\right)\,,
\end{equation}
where $c_t \equiv \kappa/(e^{\lambda_1t}-1)$ and $\tilde{\mu}$ is defined as for the SQB process in Section~\ref{sect2}.

The CIR process is reduced to an~SQB process with the same parameters $\lambda_0$ and $\nu$ by means of scale and time transformation, $Y_t = e^{-\lambda_1 t} X_{s_{\lambda_1}(t)}$, where the monotonic time-transformation function $s_{\lambda_1}$ is defined by
\begin{equation} \label{timetransf} s_{\lambda_1}(t) \equiv \left\{\begin{array}{ll} t &\mbox{ if } \lambda_1=0,\\ \frac{e^{\lambda_1 t}-1}{\lambda_1}  &\mbox{ if } \lambda_1\neq 0. \end{array} \right.\end{equation}

The transition PDF for the CIR process relates to that of the SQB process as follows:
\[
    p^{(CIR)}(t;x,y) = e^{\lambda_1 t}p^{(SQB)}(s_{\lambda_1}(t);x,e^{\lambda_1 t}y).
\]

If a~reflecting boundary condition is imposed at $x=0$, or the origin is entrance, then the CIR diffusion is a conservative stochastic process. The~corresponding transition density is given by (\ref{CIRden}) with $\tilde\mu=\mu>-1$. The transition distribution of the conservative CIR model reduces to  the randomized gamma distribution of the first type. The respective SQB process admits no absorption at zero and can be simulated by the sequential method in Figure~\ref{AlgSQB1}.

Consider the case where $x=0$ is a killing boundary or exit, so the transition PDF is given by (\ref{CIRden}) with $\tilde\mu=|\mu|$, where $\mu<0$. The FHT, $\tau_0$, at zero for the CIR model is given by
      \[ \tau_0^{(CIR)} \equiv \inf\{t\,:\, Y_t =0 \} \equiv \inf\{t\,:\, X_{s_{\lambda_1}(t)} = 0\} \overset{d}{=} s_{\lambda_1}^{-1}(\tau_0^{(SQB)}), \]
where we define $s_{\lambda_1}^{-1}(\tau)=\infty$ if $\tau>s_{\lambda_1}(\infty)$.
The corresponding PDF is given by $q^{(CIR)}(x_0;\tau) = e^{\lambda_1 \tau}q^{(SQB)}(x_0;s_{\lambda_1}(\tau))$. We have that $\P\{\tau_0^{(CIR)}<\infty\} = \P\{\tau_0^{(SQB)}<s_{\lambda_1}(\infty)\}$.
%Therefore there are two cases:
%         \begin{enumerate}[(i)]
%            \item $\lambda_1\geq 0\Rightarrow s_{\lambda_1}(\infty)=\infty\Rightarrow \P\{\tau_0^{(CIR)}<\infty\}=1$,\\
%            \item $\lambda_1< 0\Rightarrow s_{\lambda_1}(\infty)=\frac{1}{|\lambda_1|}<\infty\Rightarrow \P\{\tau_0^{(CIR)}<\infty\}<1$.
%         \end{enumerate}

Clearly, the sampling of a CIR path at times $t_i,$ $i=0,1,\ldots,N$, reduces to the sampling of an SQB trajectory. The method for sampling a path and the FHT, $\tau_0$, is given as follows.
\begin{enumerate}[Step 1.]
  \item Set times $s_i=s_{\lambda_1}(t_i)$, $i=0,1,\ldots,N$.
  \item Obtain a sample path $(X_0,X_1,\ldots,X_N)$ of the SQB process at times $s_i$, $i=0,1,\ldots,N$, and the FHT, $\tau_0^{(SQB)}$, (or its approximation $\tilde\tau_0$) by using one of the algorithms in Figures \ref{AlgSQB1}--\ref{AlgSQB4}.
  \item  Set $Y_i = e^{-\lambda_1 t_i} X_i$, $i=0,1,\ldots,N$.
  \item Set $\tau_0^{(CIR)}=\left\{\begin{array}{ll}
                              s^{-1}(\tau_0^{(SQB)})  & \text{if } \tau_0^{(SQB)}<s_{\lambda_1}(\infty)\\
                              \infty & \text{otherwise}
                            \end{array} \right.$.
  \item Return $(Y_0,Y_1,\ldots,Y_N)$ and
 $\tau_0^{(CIR)}$.
\end{enumerate}

\subsection{The CEV Diffusion Model}\label{subsect4.2}
The~constant elasticity of variance (CEV) diffusion process $\{F_t\}_{t\geq 0}$ obeys the~sto\-cha\-stic differential equation $dF_t=r F_tdt+\delta F_t^{\beta +1}dW_t,$ $t\geq 0,$ $F_0>0$,
where $r ,\delta ,\beta $ are real parameters.  We assume here that $\delta>0$ and $\beta<0$.

The~boundary $F=0$ of the~state space $(0,\infty)$ is regular if $\beta<-0.5$ or exit if $-0.5\leq \beta<0$. Here we consider the~case where the~endpoint $F=0$ is a~killing boundary. The~transition PDF $p_0(t;F_{0},F)$, $F_0,F>0,$ $t>0,$ for the~CEV process $(F^{(0)}_t)$ with zero drift ($r=0$) takes the~form
\begin{equation}
p_0(t;F_{0},F)=\frac{F^{-2\beta -\frac{3}{2}}F_{0}^{\frac{1}{2}}}{\delta ^{2}|\beta |t}%
\exp \left( -\frac{F^{-2\beta }+F_{0}^{-2\beta }}{2\delta ^{2}\beta ^{2}t}%
\right) I_{\frac{1}{2|\beta |}}\left( \frac{F^{-\beta }F_{0}^{-\beta }}{%
\delta ^{2}\beta ^{2}t}\right)\,. \label{cevpdf0}
\end{equation}%
The~density $p_{0}(t;F_{0},F)$ does not integrate (with respect to $F$) to unity for $t>0$, since $F=0$ is an absorbing point.

A~drifted CEV process $F^{(r)}_t$ with $r\neq 0$ is obtained from $F^{(0)}_t$ by means of scale and time transformation:
$F^{(r)}_t=e^{r t}F^{(0)}_{s_{\lambda_1}(t)}$,  where $s_{\lambda_1}$ is given by (\ref{timetransf}) with $\lambda_1\equiv2r\beta$.
The~resulting transition density $p_r$ with $r\neq 0$ is given by
$
p_r(t;F_{0},F)=e^{-r t}p_0(s_{\lambda_1}(t);F_0,e^{-r t}F).
$

The~Monte Carlo simulation of the~CEV diffusion is based on the~reduction of it to the CIR or SQB process by using the mapping $\X(F)\equiv\frac{F^{-2\beta}}{\delta^2\beta^2}$.  There are two dual approaches:
\begin{enumerate}[(i)]
  \item First, eliminate the drift and then, by using the mapping $\X$, reduce the driftless CEV process to an~SQB process defined by $X_t=\X(F^{(0)}_t)$, $t \geq 0$, with $\lambda_0=2+1/\beta$ and $\nu=2$. Sample a path of the SQB process and then obtain a path of the driftless CEV process by applying the mapping $\F(x)\equiv(\delta^2\beta^2 x)^{-1/2\beta}$. After that, restore the drift using the time and scale transformation.
  \item By using the mapping $\X$, reduce the drifted CEV process to a CIR process defined by $Y_t=\X(F^{(r)}_t)$, with $\lambda_0=2+1/\beta$, $\lambda_1=2r\beta$, and $\nu=2$. The resulting CIR process can be obtained from an~SQB process by means of time and scale transformation. Sample a path of the CIR process and then obtain a path of the CEV model by applying the inverse mapping $\F$.
\end{enumerate}
The FHT, $\tau_0$, at zero for the CEV diffusion model is given by
      \[ \tau_0^{(CEV)} \equiv \inf\{t\,:\, F_t =0 \} \overset{d}{=} \tau_0^{(CIR)} \overset{d}{=} s_{\lambda_1=2r\beta}^{-1}(\tau_0^{(SQB)}). \]

Notice that if a~reflecting boundary condition is imposed at $F=0$ when $\beta<-0.5$ (or $\beta>0$ and hence $F=0$ is entrance), then the CEV diffusion is a conservative stochastic process. The~corresponding transition density (for the~case with $\beta<-0.5$) is given by (\ref{cevpdf0}) with the replacement $I_{\frac{1}{2|\beta|}}\to I_{\frac{1}{2\beta}}$. By analogy with the CIR model without absorption at zero, the transition distribution of the conservative CEV model reduces to  the randomized gamma distribution of the first type, hence the algorithm in Figure~\ref{AlgSQB1} is applied.

\subsection{Diffusion Canonical Transformation}\label{subsect4.3}
Several families of analytically solvable diffusions
can be derived from known underlying diffusion processes. We refer to
this construction as the \textit{``diffusion canonical
transformation''} methodology (see \cite{CM07,CM08,CM09} for details).

Let us start with a one-dimensional time-homogeneous regular diffusion $(X_t)_{t\geq 0}\in\I\equiv (l,r)$, $-\infty\leq l<r\leq\infty$,
defined by its infinitesimal generator:
$(\G\, f)(x)\equiv {1\over
2}\nu^2(x)f^{\prime\prime}(x) + \lambda(x)f^\prime(x)$.
The functions $\lambda$ and $\nu$ denote, respectively, the
(infinitesimal) drift and diffusion coefficients of the process.
Consider two linearly independent
\textit{fundamental solutions} $\varphi^+_s$ and $\varphi^-_s$ of
the differential equation $(\G\, \varphi)(x) = s
 \varphi(x),$ $s\in \C,$ $x\in\I,$
such that for real values $s=\rho>0$ the
solutions $\varphi^+_\rho$ and $\varphi^-_\rho$ are
respectively increasing and decreasing functions of $x$ (see,
e.g.,~\cite{BS02}).

Let us introduce another diffusion $(X^{(\rho)}_t)_{t\ge 0}\in \I$ with generator
 \begin{equation} \label{Generator_rho} (\G^{(\rho)}\,f)(x)\equiv\frac{1}{2}\nu^2(x)f^{\prime\prime}(x)
 +\left(\lambda(x)+\nu^2(x)\frac{u'_\rho(x)}{u_\rho(x)}\right)f^{\prime}(x)\,,
 \end{equation}
where a {\it strictly positive}  function
$u_\rho(x),$ $\rho>0$, is a linear combination of~$\varphi^\pm_\rho$:
$u_\rho(x)=q_1 \varphi^+_\rho(x)
+ q_2 \varphi^-_\rho(x),\; q_{1,2}\geq 0, \;q_1+q_2>0.
$
A transition density $p_X^{(\rho)}$ for the $X^{(\rho)}$-diffusion is then related to  a transition density $p_X$ for the $X$-diffusion as follows:
\begin{equation}
p_X^{(\rho)}(t;x_0,x)=e^{-\rho
t}\frac{u_\rho(x)}{u_\rho(x_0)}
p_X(t;x_0,x),\;x,x_0\in\I\,,\;t>0\,. \label{urho}
\end{equation}

Now we consider an $F$-diffusion $\{F_t \equiv \F
(X^{(\rho)}_t), t\ge 0\}$ defined by strictly monotonic real-valued mapping
$F=\F (x)$ with $\F ', \F ''$ continuous on $\I$ and
having infinitesimal generator
$(\mathcal{G}_{\sf F} h)(F) \equiv \frac{1}{2}\,\sigma^2(F)
h''(F) + rFh'(F)$, where $F\in\I_{\sf F}=\left(\min\{\mathsf{F}(l+),\mathsf{F}(r-)\},\max\{\mathsf{F}(l+),\mathsf{F}(r-)\}\right)$, and $r$ is a real constant so that $\rho+r>0$.

The transition PDF $p_F$ for an $F$-diffusion $(F_t)_{t\ge 0}$ is related to
the transition PDF for the underlying $X$ (or $X^{(\rho)}$) diffusion as follows:
\begin{equation} \label{PrKernelF}
  p_F(t;F_0,F) = \frac{\nu(\X (F))}{\sigma(F)}
  \frac{u_\rho\left(\X (F) \right)}{u_\rho\left(\X (F_0) \right)}
  e^{-\rho t} p_X(t;\X (F_0),\X (F))
  \,.
\end{equation}
Here $\X \equiv\F^{-1}$ is the inverse
map. $\F$ admits the general quotient form:
 \begin{equation} \label{FMAP} \F (x) = \frac{c_1 \varphi^+_{\rho+r}(x)+c_2
 \varphi^-_{\rho+r}(x)}{q_1 \varphi^+_{\rho}(x)+q_2
 \varphi^-_{\rho}(x)}\equiv \frac{{v}_{\rho+r}(x)}{u_\rho(x)}
 \end{equation}
where $c_1$ and $c_2$ are real constants. For a full classification of strictly monotonic maps of the form (\ref{FMAP}) see \cite{CM09}.
The diffusion coefficient function is
\begin{equation} \label{sigmaF}
 \sigma(F)=\frac{\nu(x)|W(x)|}{u^2_\rho(x)}\,,\quad x=\X (F)\,,\quad F\in\I_{\sf F}\,,
\end{equation}
where we define the Wronskian $W(x) \equiv u_\rho(x)v'_{\rho+r}(x)-u'_\rho(x)v_{\rho+r}(x)\,.$

In the next two subsections we present two examples of hypergeometric diffusions. The concluding subsections gives a general simulation algorithm.

\subsection{The Bessel-$\mathsf{K}$  Diffusions} \label{sect4.4}
Here we specifically
consider a $4$-parameter Bessel $\mathsf{K}$-family arising from an underlying
($\lambda_0$-dimen\-sional) squared Bessel process with a positive index $\mu$. We use the generating function
$u_\rho(x)=\varphi^-_\rho(x)=x^{-\mu/2}K_{\mu}\left(2\sqrt{2\rho x}/\nu\right)$ and the mapping:
\begin{equation} \label{FmapBesK}
  \F(x)=
     c\displaystyle\frac{I_{\mu} \left( 2\sqrt{2(\rho+r) x}/\nu\right)}{K_{\mu} \left(2\sqrt{2\rho x}/\nu\right)},
\end{equation}
where $c$, $\rho,$ $\nu$, and $\mu$ are independently adjustable
positive parameters, and $r>-\rho$ is a real constant. The functions $I$ and
$K$ denote the modified Bessel functions of the first and second kind, respectively. (see \cite{AS72} for definitions and properties).

 The function $\F(x)$
(and the respective inverse $\X(F)$) maps $x\in(0,\infty)$ and $F\in(0,\infty)$ into one another. The
transformation (\ref{FmapBesK}) hence leads to a family of processes
$(F_t)\in (0,\infty)$ with the diffusion coefficient function
\begin{equation}\label{SigmaBesK}
  \sigma(\F(x))= c\sqrt{2}\left(
     \textstyle\frac{\sqrt{\rho}\,I_\mu\left(\frac{2}{\nu}\sqrt{2(\rho + r) x}\right)K_{\mu+1}\left(\frac{2}{\nu}\sqrt{2\rho x}\right)}
{K_{\mu}^2\left(\frac{2}{\nu}\sqrt{2\rho x}\right)}
+ \textstyle\frac{\sqrt{\rho + r}\,I_{\mu+1}\left(\frac{2}{\nu}\sqrt{2(\rho + r)x}\right)}{K_{\mu}\left(\frac{2}{\nu}\sqrt{2\rho x}\right)} \right)
\end{equation}

\begin{lemma}[Campolieti and Makarov, \cite{CM07,CM09}] \label{lemK}
The processes of the Bessel $\mathsf{K}$-family obeying the SDE
$dF_t=rF_tdt+\sigma(F_t)dW_t$ with (\ref{FmapBesK})--(\ref{SigmaBesK}) have the following boundary classification:
the boundary $F=0$ is exit if $\mu\geq 1$ or is a regular killing boundary if
$0<\mu<1$; the boundary $F=\infty$ is non-attracting natural. Moreover, the discounted process $(e^{-rt} F_t)_{t\ge 0}$ is a martingale. The transition PDF $p_F$ is given by (\ref{PrKernelF}) with $\nu(x)=\nu\sqrt{x}$, and $\sigma$ and $p_X$ respectively specified by (\ref{SigmaBesK}) and (\ref{PrkernelSQB}).
\end{lemma}

The density, $q(F_0;\tau)$, for the FHT at the origin for a Bessel-$\mathsf{K}$ process started at $F_0>0$
is readily derived by using equation (\ref{FHTPDFlim}), giving
the generalized inverse Gaussian distribution:
\begin{equation}
q(F_0;\tau) =
\displaystyle\frac{\big(2x_0/\rho\nu^2\big)^{\mu/2}}
{2\, K_\mu\big(2\sqrt{2\rho x_0}/\nu\big)}\,
\tau^{-\mu - 1} e^{-\rho\tau -
2x_0/\nu^2\tau},\,\,\,\tau > 0,\;x_0 = \X(F_0). \label{GIG}
\end{equation}

%The generating function $u_\rho(x)$ is a linear combination of
%\begin{equation} \label{SQBfund}
%\varphi^+_\rho(x) =
%x^{-\mu/2}I_{\mu}\big(2\sqrt{2\rho x}/\nu\big) \mbox{ \ and \ }
%\varphi^-_\rho(x) = x^{-\mu/2}K_{\mu}\big(2\sqrt{2\rho x}/\nu\big)\,,
%\end{equation}
%where $I_\mu(z)$ and $K_\mu(z)$ are the modified Bessel functions
%(of order $\mu$) of the first and second type, respectively (for
%definitions  and properties see \citet{AS72}). The pair $\varphi^\pm_s(x)$
%satisfies (\ref{wronskian}) where $w_s=1/2$.

\subsection{The Confluent-$\mathcal{U}$ Diffusions} \label{sect4.5}
The confluent hypergeometric family of $F$-diffusions arises from an underlying CIR process with $\mu>0$. Here we specialize to the \emph{confluent}-$\U$ family with
generating function $u_\rho(x)=\varphi^-_\rho(x)=\U(\upsilon,\mu+1,\kappa x)$
and mapping
\begin{equation} \label{FmapCIRU}
\F(x)=c\frac{\M(\upsilon+\frac{b}{\lambda_1},\mu+1,\kappa
x)}{\U(\upsilon,\mu+1,\kappa x)},
\end{equation}
where  $\upsilon \equiv \frac{\rho}{\lambda_1}$, $\mu\equiv\frac{2\lambda_0}{\nu^2}-1$, $\kappa \equiv \frac{2\lambda_1}{\nu^2},$ and $c$ are
arbitrary positive constants, and $r>-\rho$. The confluent hypergeometric functions $\M$ and
$\U$ are two linearly independent solutions to Kummer's differential equation (see \cite{AS72} for definitions and properties).

The function $\F(x)$ maps $x\in(0,\infty)$ onto
$F\in(0,\infty)$ and is monotonically increasing. This
transformation leads to a family of processes $(F_t)\in (0,\infty)$
with the diffusion coefficient function
\begin{equation} \label{SigmaCIRU}
    \hspace{-0mm}\sigma(\F(x))= c\kappa\nu\sqrt{x}\left(
    \textstyle\frac{\upsilon \,\M\left({\rho+r\over \lambda_1},\mu+1,\kappa x\right)\U\left(\upsilon + 1 ,\mu+2,\kappa x\right)} {\U^2\left(\upsilon,\mu+1,\kappa x\right)} +
    \textstyle\frac{({\rho+r\over \lambda_1})\,\M\left({\rho+r\over \lambda_1}+1,\mu+2,\kappa x\right)}{(\mu + 1)\,\U\left(\upsilon,\mu+1,\kappa x\right)}\right)\hspace{-3mm}
\end{equation}

\begin{lemma}[Campolieti and Makarov, \cite{CM07,CM09}] \label{lemU} The processes of the confluent
 $\mathcal{U}$-family solving the SDE $dF_t=rF_tdt+\sigma(F_t)dW_t$ with (\ref{FmapCIRU})--(\ref{SigmaCIRU}) have the same boundary classification as that for the Bessel-$\mathsf{K}$ in Lemma~\ref{lemK}. Moreover, the discounted process $(e^{-rt} F_t)_{t\ge 0}$ is a martingale. The transition PDF $p_F$ is given by (\ref{PrKernelF}) with $\nu(x)=\nu\sqrt{x}$, and $\sigma$ and $p_X$ respectively specified by (\ref{SigmaCIRU}) and (\ref{CIRden}).
\end{lemma}

The density for the first-hitting time at the origin, $q(F_0;\tau)$, for a confluent-$\U$ process started at $F_0>0$ is
\begin{equation} \label{Tricomi} q(F_0;\tau) = \left|{\mathcal T}^\prime(\tau)\right| \frac{e^{-\kappa x_0 {\mathcal T}(\tau)}( {\mathcal T}(\tau))^{\upsilon-1}(1+{\mathcal T}(\tau))^{\mu-\upsilon}}{\U(\upsilon,\mu+1,\kappa x_0)\Gamma(\upsilon)},\;\tau >0,
\end{equation}
where $x_0 = \X(F_0)$ and we use the time change ${\mathcal T}(\tau) \equiv \displaystyle\frac{e^{-\lambda_1 \tau}}{1-e^{-\lambda_1 \tau}}$.
The latter function in (\ref{Tricomi}) is known as a Tricomi exponential PDF (see \cite{Fitz02}) given by $p({\mathcal T}) = \displaystyle\frac{e^{-z{\mathcal T}}{\mathcal T}^{a-1}(1+{\mathcal T})^{b-a-1}}{\Gamma(a)\,\U(a,b,z)},$ ${\mathcal T}>0,$ where $a=\upsilon,$ $b=\mu+1,$ $z=\kappa x_0,$.
It integrates to unity thanks to the integral representation of $\U$ (see \cite{AS72}).

\subsection{Simulation of $F$-Diffusions} \label{sect4.6}
We generalize the sampling algorithms for an SQB process presented in Figure~\ref{AlgSQB3} and Figure~\ref{AlgSQB4}. Within that approach a path is sampled conditionally on the FHT at zero. The Bessel-\textsf{K} and confluent-$\mathcal{U}$ diffusion models are both absorbing at zero and have the first-hitting time distribution in analytically closed-form. For a sampling algorithm we only need to obtain the distribution of the respective bridge process. In doing, so we use one important observation that the distribution of an $F$-diffusion bridge is reduced to the distribution of a bridge of the respective underlying diffusion  (e.g. the  Bessel and CIR bridges).

By applying the analogue of formula (\ref{BRGPDFfull}) for an $F$-diffusion with PDF $p_F(t;F_0,F)$ in place of the PDF $p(t;x,y)$, and using the representation (\ref{PrKernelF}), we have the following expression for the bridge PDF of an $F$-diffusion with $F_{t_1}$ and $F_{t_2}$ tied at $F_1$ and $F_2$ respectively:
\begin{align}
  \label{BrdgDenF} b_F(t_1,t_2,t;F_1,F_2,F) % &\equiv \frac{p_F(t-t_1;F_1,F)p_F(t_2-t;F,F_2)}{p_F(t_2-t_1;F_1,F_2)} \\
   &= \frac{\nu(\X(F))}{\sigma(F)}\, b_X^{(\rho)}(t_1,t_2,t;\X(F_1),\X(F_2),\X(F)) \\
  \nonumber &= \frac{\nu(\X(F))}{\sigma(F)}\, b_X(t_1,t_2,t;\X(F_1),\X(F_2),\X(F))
%  \mbox{} =
%  \frac{\nu(\X(F))}{\sigma(F)}\,\frac{p_X^{(\rho)}(t-t_1;\X(F_1),\X(F))
%  p_X^{(\rho)}(t_2-t;\X(F),\X(F_2))}{p_X^{(\rho)}(t_2-t_1;\X(F_1),\X(F_2))}\\[5mm]
%  \mbox{} = \displaystyle
%  \frac{\nu(\X(F))}{\sigma(F)}\,\frac{p_X(t-t_1;\X(F_1),\X(F))
%  p_X(t_2-t;\X(F),\X(F_2))}{p_X(t_2-t_1;\X(F_1),\X(F_2))}.
\end{align}
where $b_X$ and $b_X^{(\rho)}$ denote the bridge PDFs of the diffusions $(X_t)$ and $(X^{(\rho)}_t)$, respectively.
Here, after plugging (\ref{PrKernelF}) in the formula of the bridge PDF $b_F$, we first cancel  Jacobians $\frac{\nu}{\sigma}$ and then cancel Doob's factors of the form $e^{-\rho t}\frac{u_\rho(y)}{u_\rho(x)}$.
If follows from (\ref{BrdgDenF}) that an $F$-diffusion bridge is obtained by applying the mapping function $\F$ to the bridge process for the underlying diffusion $(X_t)$ with $X_{t_1}$ and $X_{t_2}$ tied at $\X(F_1)$ and $\X(F_2)$ respectively. For example, in the particular case of the Bessel-\textsf{K} diffusion when the underlying process $(X_t)$ is a squared Bessel process, the $F$-bridge is just a nonlinear transformation of a standard Bessel bridge.

Our primary goal is to sample a path skeleton $(F_0,F_1,\ldots,F_N)$, $F_i\equiv F_{t_i}$ of an $F$-diffusion at times $t_i$, $i=0,1,\ldots,N$, $0=t_0<t_1<\cdots<t_N=T$, for a given initial condition $F_{t=0}=F_0$. The simulation scheme based on the bridge distribution is as follows:
  \begin{enumerate}[Step 1.]
    \item Sample the FHT, $\tau_0$, from the GIG or exponential Tricomi distribution for the Bessel-\textsf{K}  or Confluent-$\mathcal{U}$ model, respectively.
    \item Obtain a sample path $(X_0,X_1,\ldots,X_N)$ of the respective underlying process (the SQB or CIR diffusion) conditional on $X_0=\X(F_0)$ and $X_{\tau_0}=0$.
    \item Apply the respective mapping function $\F$ to obtain a sample path of the $F$-diffusion model: $F_i=\F(X_i)$, $i=1,2,\ldots,N$.
  \end{enumerate}
The main result is that this simulation scheme allows us to avoid a direct sampling from complicated transition probability distributions.

Let us present an alternative approach from \cite{CM08} to computing mathematical expectations of path functionals of the form $Q\equiv\E[f(F_1,F_2,\ldots,F_N)|F_0]$ for $F$-diffusion. By using a path integral approach, the expected value of such a path functional can be represented as a multivariate integral:
\[
    Q = \int_{\R^N} f(\F(x_1),\ldots \F(x_N)) e^{-\rho T}\frac{u_\rho(x_N)}{u_\rho(x_0)} \prod\limits_{k=1}^N p_X(t_k-t_{k-1};x_{k-1},x_k) dx_1\cdots dx_N.
\]
The integral above may be estimated by the Monte Carlo method. The underlying diffusion is simulated by sampling from the exact transition probabilty distribution.  The resulting unbiased estimator $\xi$ of the path integral $Q$ takes the form:
\[ \xi = f(\F(X_1),\ldots \F(X_N)) e^{-\rho T}\frac{u_\rho(X_N)}{u_\rho(X_0)},
\]
where the path $(X_0,X_1,\ldots,X_N)$ is sampled by using one of the algorithms from Section~\ref{sect3}. Notice that we cannot use the Euler method (or any other approximation method, which does not guarantee the positiveness of the approximation process) since the estimator $\xi$ is infinite if $X_N=0$. Using large and small argument asymptotics of the Bessel function $K$ and Kummer function $\U$, we obtain that the variance of $\xi$ is finite if  $\mu<1$. Notice that the use of an exact simulation method allows us to lift this restriction.

\section{Simulation Study} \label{sect5}
\subsection{Simulation of Randomizers} \label{sect5.1}
The three discrete probability distributions used in the construction of randomized gamma distributions are all log-concave and unimodal as is stated below.
\begin{lemma}
Let $Y$ be a Poisson, Bessel or incomplete gamma random variable. The distribution of $Y$ is log-concave. That is, the ratio
$\P\{Y=n+1\}/\P\{Y=n\}$ is decreasing in $n$. Furthermore, the distribution of $Y$ is unimodal and has a unique mode or two modes at consecutive integers. Moreover, one mode is always located at $m=\left\lfloor\lambda\right\rfloor$ for the Poisson distribution, at $m=\left\lfloor(\sqrt{b^2+\theta^2}-\theta)/2\right\rfloor$ for the Bessel distribution and at $m=\max(0,\lfloor \lambda -\theta\rfloor)$ for the incomplete gamma distribution.
\end{lemma}
\begin{proof} See \cite{Devr01,CM07} for the proof for the Bessel and incomplete gamma distributions, respectively. \end{proof}

To generate a Bessel or incomplete gamma random variate, we can use a generic acceptance-rejection (A-R) method from \cite{Devr87} stated below without proof.
\begin{lemma}[Devroye, \cite{Devr87}]
For any discrete log-concave distributions with mode at $m$, we
have, for all $n\ge 0$: $p_n\le p_m\min\left\{1,e^{1-p_m|n-m|}\right\}.$
\end{lemma}
%As is seen from the previous lemma, for all $k\ge -m$ and $k-1/2\le x\le k+1/2$, $p_{m+k}\le g(x)$ holds, where $g(x)=p_m\min\left\{1,e^{1-p_m(|x|-1/2)}\right\}.$ So, the function $g$ is a dominant function for any discrete log-concave distribution with mode at $m$ and with the probability vector $\{p_n\}_{n\ge 0}$. Observe that $g$ is a mixture of a rectangular function on $[-w/p_m,w/p_m]$ and two exponential tails outside $[-w/p_m,w/p_m]$, where $w=1+p_m/2$. When $g$ is used in the rejection algorithm, the rejection constant (expected number of iterations before succeeding) is $\int_{-\infty}^\infty g(x)dx=2w+2=4+p_m\le 5$ (see \cite{Devr87} for details).

As an alternative sampling method we use the inversion method by chop-down  search (C-D-S) from the mode $m$. Such a sampling method for a discrete distribution with probabilities $\{p_k\}_{k\geq 0}$ is based on the numerical inversion of the CDF $F$ by the formula
$F^{-1}(u)=\mathrm{arg}\min\{n\geq 0 \mid u-\sum_{k=0}^n p_k< 0\},\,u\in[0,1].$

It is well known that the computational cost of such a method has the lowest possible value if and only if the vector of discrete probabilities is arranged in increasing order. Instead of the preliminary computation of probabilities followed by sorting of them, we start the search algorithm at the mode $m$ and then successively calculate probabilities of values to the left and to the right of the mode choosing the largest one. Notice that probability $p_m$ need only be computed once, and that other probabilities can be obtained by using simple recurrences.

We now present some numerical results comparing the two methods of simulation of P($\lambda$), Bes($\theta$,b), and I$\Gamma(\theta,\lambda)$ random variables. For each of the two methods, one million values are sampled. For simulation of each of the Poisson random variables, the parameter $\lambda$ is allowed to vary as a continuous uniform random variable. For the two parameter Bessel and incomplete gamma distributions, the first parameter is allowed to vary as a continuous uniform random variable while the second parameter is held constant. Then the procedure is repeated by allowing the second parameter to vary while the first one is held constant. Results of these tests are given in Table~\ref{Table:SampRand}.
\begin{table}[h]
{\footnotesize
	\caption{Comparison of the acceptance-rejection and chop-down search methods for the Poisson, Bessel, and incomplete gamma distributions. \label{Table:SampRand}}
	\begin{center}
	\begin{tabular}{ l l c c c c}
                \multicolumn{2}{l}{Distribution}  &  \multicolumn{2}{c}{A-R Method} & \multicolumn{2}{c}{C-D-S Method}\\
		& & Time & No. of Iter. & Time &  No. of Iter. \\ \hline
		$\mathrm{P}(\lambda)$ & $\lambda \sim \mathrm{U}(0,1000)$ & 189.9 & 2.6 & 35.2 & 34.2\\ \hline
		$\mathrm{Bes}(\theta,b)$ & $\theta \sim \mathrm{U}(0,1000)$, $b = 10$ &  220.6 & 1.6 & 100.4 & 1.1 \\ \hline
		$\mathrm{Bes}(\theta,b)$ & $\theta = 10$, $b \sim \mathrm{U}(0,1000)$ & 414.1 & 4.0 & 103.3 & 17.3 \\ \hline
		$\mathrm{I}\Gamma(\theta,\lambda)$ & $\theta \sim \mathrm{U}(0,100)$, $\lambda = $10 &  336.6 & 3.7 & 51.1 & 1.8 \\ \hline
		$\mathrm{I}\Gamma(\theta,\lambda)$ & $\theta = 10, \lambda \sim \mathrm{U}(0,1000)$ & 363.7 & 3.9 & 51.4 & 10.8\\ \hline
	\end{tabular} \\
	\end{center}
}
\begin{spacing}{0.5}
{\footnotesize \textit{Note}: Time in seconds and average number of iterations for the simulation of $10^6$ random variables from the Poisson, Bessel, and incomplete gamma distributions using the acceptance-rejection (A-R) and chop-down search (C-D-S)  methods.}
 %The following data is used: Test 1. $\mathrm{P}(\lambda)$: $\lambda \sim \mathrm{U}(0,1000)$. Test 2. $\mathrm{Bes}(\theta,b)$: $\theta \sim \mathrm{U}(0,1000)$, $b = 10$. Test 3. $\mathrm{Bes}(\theta,b)$: $\theta = 10$, $b \sim \mathrm{U}(0,1000)$. Test 4. $\mathrm{I}\Gamma(\theta,\lambda)$: $\theta \sim \mathrm{U}(0,100)$, $\lambda = $10. Test 5. $\mathrm{I}\Gamma(\theta,\lambda)$: $\theta = 10, \lambda \sim \mathrm{U}(0,1000)$.
\end{spacing}
\end{table}

Table 1 shows that the chop-down search method from the mode is significantly faster than the acceptance-rejection technique for generating random variables in every case and is a much better choice for simulating random variables when it can be implemented.

\subsection{Comparison of Sampling Schemes for the SQB Process} \label{subsect5.2}
In this section we aim to compare the following three sampling schemes.
\begin{enumerate}[1)]
  \item Sequential sampling conditional on the FHT $\tau_0$ with the use of the randomized gamma distribution of the first kind.
  \item Bridge sampling conditional on the FHT $\tau_0$ with the use of the randomized gamma distribution of the second kind.
  \item Unconditional sequential sampling with the use of the randomized gamma distribution of the third kind.
\end{enumerate}
We start by sampling multiple paths of the SQB process over a discretized partition of a time interval $[0,T], 0=t_0<t_1<\cdots<t_N=T,$ using one of the three methods just mentioned. Then we average these sample paths in order to approximate the mean of the SQB process. To study the sampling algorithms, we compare our sample means to the true mean of the SQB process as well as the time required to simulate a set number of sample paths of the process.

For calculation of the mean of the SQB process, we use the formula
\begin{equation*} \E[X_t] = \frac{x_0+\lambda_0 t}{\Gamma(\vert \mu \vert)}\gamma\left(\vert \mu \vert,\frac{x_0}{2t} \right)  + \frac{x_0}{\Gamma(\vert \mu \vert)}\left(\frac{x_0}{2t}\right)^{\vert \mu \vert-1} \exp\left(-\frac{x_0}{2t}\right),
\end{equation*}
which is valid for $\mu < 0$ and $\nu=2$. The expression is derived by considering the moment generating function of the SQB process at time $t$ and using the small asymptotics of the Bessel $I$ function.

To this end, we look at the largest amount by which the sample mean, $\bar{\mu}_t$, differs from the true mean, $\mu_t \equiv \E[X_t]$, (the maximum absolute error) at times $t_i, i=0,1,\ldots,N$, given by $\displaystyle \max_{i=0,1,\ldots,N} \vert \mu_{t_i} - \bar{\mu}_{t_i} \vert$. We will also examine the largest sample standard deviation of the process, given by $\displaystyle \max_{i=0,1,\ldots,N} \bar{\sigma}_{t_i}/\sqrt{n}$, where $n$ is the sample size. After simulating one million sample paths for each of the three sample schemes and averaging them, we obtain the data shown in Table~\ref{Table:SQBsampling}.
\begin{table}[h]
{\footnotesize
	\caption{Comparison of sampling schemes for the SQB process. \label{Table:SQBsampling}}
	\begin{center}
	\begin{tabular}{c | ccc | ccc | ccc}
              \multicolumn{1}{c}{} & \multicolumn{3}{c}{\textbf{Scheme 1}} & \multicolumn{3}{c}{\textbf{Scheme 2}} & \multicolumn{3}{c}{\textbf{Scheme 3}}\\
  $\mu$ & Time  & MAE & MST & Time  & MAE & MST & Time  & MAE & MST\\ \hline
  -0.25& 1741 & .00244 & .00271 & 5044 & .00297 & .00270 & 2501 & .00534 & .00270 \\ \hline
  -0.5 & 1600 & .00111 & .00252 & 4462 & .00191 & .00251 & 2280 & .00158 & .00251 \\ \hline
  -1.5 & 953. & .00193 & .00144 & 2614 & .00240 & .00145 & 1406 & .00149 & .00145 \\ \hline
	\end{tabular} \\
	\end{center}
}
\begin{spacing}{0.5}
{\footnotesize \textit{Note}: Time in seconds, maximum absolute error (MAE), and the maximum standard deviation (MST) taken from the average of $10^6$ sample paths of the SQB process using sampling schemes 1, 2, and 3 respectively for varying values of $\mu$. For all three choices of $\mu$, we set $X_0 = 1,$ $T = 1$, $\nu=2$, and the partition of $[0,T]$ to be $0, \frac{1}{32}, \frac{2}{32}, \ldots ,1.$
%For $\mu=-0.25$, we choose $\lambda_0 = 3/2, \nu = 2$. For $\mu=-0.5$, we choose $\lambda_0 = 1, \nu = 2$. For $\mu=-0.5$, we choose $\lambda_0 = -1, \nu = 2$.
}
\end{spacing}
\end{table}
From this data, we can see that sampling scheme 1  is the fastest one. Scheme 2 is much slower than schemes 1 and 3 since it involves sampling from the Bessel distribution.

\subsection{Sampling from the GIG and Tricomi Exponential Distributions} \label{sect5.3}
A common approach to sampling from a nonstandard probability distribution is to use an~acceptance-rejection method. This approach is employed in \cite{ATK82} and \cite{Fitz02} for sampling from the GIG and Tricomi exponential distributions, respectively. If the parameters of a probability distribution remain constant, then a much faster sampling technique is the one that is based on the numerical inversion of a distribution function. To sample from a continuous CDF $F$ by using the inverse transform method, we generate a uniformly distributed on $(0, 1)$ random variable $U$ and then set $X=F^{-1}(U)$, where $F^{-1}$ the inverse of $F$.

In cases where the inverse of $F$ can not be expressed in closed-form, the inverse transform relies on numerical approximation. A root-finding method such as Newton's method or the bisection method can be applied to solve equation $F(X)=U$, $U\in(0,1)$. A faster approach is to compute the CDF on a fine mesh and then approximate the inverse of the CDF by some simpler functions. The simplest method is to use a piece-wise linear interpolation. In \cite{HORM03} a fast and efficient variate generation method is proposed. In that method, the inverse CDF $F^{-1}$ is approximated by the Hermite interpolation functions. For a given partition $l=x_0<x_1<\cdots<x_n=r$ of the support $(l,r)$ of a CDF $F$, the distribution function is computed by either integrating the density function on each subinterval $(x_{k-1},x_k)$, or by employing an ODE solver, since the CDF $F$ solves a simple ODE $F'(x)=f(x)$, $x\in(l,r)$, $F(l)=0$, where $f$ is the respective PDF.

\subsection{Path-Dependent Options} \label{subsect5.4}
This section reviews some discretely-monitored path-dependent options that will be used for pricing options in the following subsection. First, we assume that we have sampled a path of a asset price process $(F_t)$ over a discrete time partition, $\mathbf{T}=\{t_i \}_{i=0,1,\ldots, N}$, of the time interval $[0,T],$ $T>0$. Let the values of process $(F_t)$ at time points $t = t_i$ be denoted by $F_i$, for all $i=0,1,\ldots, N$.

The payoff function of an Asian-style option depends on the arithmetic average of the underlying asset values: $A_N = \frac{1}{N}\sum_{i=1}^N F_i$. For an \textit{average price call option}, the payoff to the option holder at time $T$ is ($A_N-K)_+$ where $K$ is the strike price and $(x)_+ \equiv \max(x,0)$. The \textit{average price put option} is defined similarly. Its payoff at time $T$ is ($K-A_N)_+$.

The second type of path-dependent options we will price are lookback options. In this case, the payoff functions depend on the maximum, $M_N = \displaystyle \max_{i=0,1,\ldots,N} F_i$, or the minimum, $m_N= \displaystyle \min_{i=0,1,\ldots,N} F_i$, values of the underlying asset price attained during the option's life, $[0,T]$. A \textit{standard lookback call} gives the right to buy at the lowest price recorded during the options life. Hence, the payoff to the holder at time $T$ is $F_N - m_N$. A \textit{standard lookback put} gives the right to sell at the highest price recorded during the options life. Thus, the payoff at time $T$ is $M_N - F_N$.

\subsection{Pricing Path-dependent Options under Nonlinear Volatility Models} \label{subsect5.5}
In this section we present some numerical results regarding pricing Asian and lookback options under the CEV, Bessel-$\mathsf{K}$ and Confluent-$\mathcal{U}$ families of diffusions using Monte-Carlo algorithms based on generating from randomized Gamma distributions. Specifically, we look at a plain sequential Monte-Carlo sampling method (MCM) and a randomized quasi Monte-Carlo method (RQMCM) which uses digital scrambling via a Sobol's sequence for the randomization. For the Bessel-$\mathsf{K}$ and Confluent-$\mathcal{U}$ models, we also use the weighted method (MCMW) described in Subsection~\ref{sect4.6}. One million simulations are completed for each payoff function and are then averaged to get the final option pricing results. For the RQMC method, these $10^6$ simulations correspond to 100 randomizations and $10\,000$ simulations per randomization.

In the tests that follow, we fix the value of the annual local volatility function $\sigma_{loc}(S_0) = 0.25$ at the initial asset price $S_0 = 100$. The strike price is $K = 100$. The interest rate is $r = 0.02$ per annum and all options have six months to expiration: $T = 0.5$. The number of asset price observations is $N = 128$.
First we look at pricing under the CEV model. For the CEV model, $\sigma_{loc}(S_0) = \delta S_0^\beta$. Typical observed values of the CEV elasticity parameter $\beta$ are strongly negative so we choose $\beta = -2$. Then we choose the parameter $\delta$ so that it satisfies $\delta F_0^\beta = 0.25$. This yields $\delta = 2500$.
Next we consider the Bessel-$\mathsf{K}$ subfamily of diffusions. To ensure that $\sigma_{loc}(F_0) = 0.25$ the following parameters are used: $\rho = 0.001$, $r = 0.02$, $c = 154.4870$, $\mu = 0.25$, and $\nu = 2$.
The last pricing model considered here is the Confluent-$\mathcal{U}$ family of diffusions. Specifically, we examine the case where $c = 788.3679$, $\rho = 0.001$, $\lambda_1 = 0.0009$, $\mu = 0.25$, and $\nu = 2$. Table~\ref{Table:OptionPrices} contains option pricing results corresponding to these models. The prices reported are obtained using the RQMC method. Table~\ref{Table:CostReduction} reports the computational cost of pricing the average price Asian call using the three methods.

\begin{table}[htb!]
{\footnotesize
	\caption{Pricing path-dependent options under the three models using the RQMC method. The value of the sample standard error is given after the $\pm$ sign.  \label{Table:OptionPrices}}
	\begin{center}
	\begin{tabular}{l r@{$\pm$}l r@{$\pm$}l r@{$\pm$}l r@{$\pm$}l}
	\textbf{Model} & \multicolumn{2}{c}{{\bf Asian Call}} & \multicolumn{2}{c}{{\bf Asian Put}} & \multicolumn{2}{c}{{\bf Lookback Call}} & \multicolumn{2}{c}{{\bf Lookback Put}} \\
\hline
CEV  & 4.30237 & .00081 & 3.80260 & .00160 &  14.55220 & .00255 &  12.09087 & .00300\\ Bessel-\textsf{K} & 4.28605 & .00049 & 3.79717 & .00033 & 13.15557 & .00113 & 13.23640 & .00081 \\
Confluent-$\mathcal{U}$ & 4.28724 & .00049 & 3.79922 & .00032 & 13.31158 & .00093  & 13.11594 & .00084\\
\hline
	\end{tabular}
	\end{center}
}
\end{table}

\begin{table}[htb!]
{\footnotesize
	\caption{Computational cost of pricing the average price Asian call option.\label{Table:CostReduction} }
	\begin{center}
	\begin{tabular}{l l r r r r}
	\textbf{Model} &  Method &  Smpl.Var., $\bar{\sigma}^2$ & Time (sec) & Cost, $\bar{\sigma}^{2}T$ & Relat. Cost\\
\hline
	{\bf CEV }	
    &	MCM   & 32.574 & 7438 & 242296 & 52.4\\
	&	RQMCM  & \phantom{3}0.065 & 70762 & 4622 & 1.0 \\
\hline
	{\bf Bessel-\textsf{K}}
	&	MCMW &  33.044 & 33291 & 1100088 & 52.6\\
	&	MCM   & 41.830 & 10506 & 439444 & 21.0\\
	&	RQMCM  &  \phantom{3}0.235 & 89029 & 20895 & 1.0\\
\hline
	{\bf Confluent-$\mathcal{U}$}
	&	MCMW    &  31.636 & 33312 & 1053853 & 49.4\\
	&	MCM    &  40.801 & 10174 & 415122 & 19.5 \\
	&	RQMCM  & \phantom{3}0.238  & 89715 & 21308 & 1.0 \\
\hline
	\end{tabular} \\
	\end{center}
}
\end{table}

As seen in Table~\ref{Table:CostReduction}, the RQMC method offers a clear improvement in reductive cost over the plain MC method. On the other hand, the weighted method offers no improvement in cost at all, mostly due to its relatively large computational time. The extra time required for the weighted method is partly due to the computation of special functions in the weight. It could also be attributed to sampling more points in each of the sample paths for a price process $(F_t)$. When conditioning on the FHT $\tau_0$ and sampling at time $t$, we check first whether $t \geq \tau_0$. If $t \geq \tau_0$ we do not have to sample from any probability distributions since $F_t = 0$. When using the weighted method we are looking at the case with no absorption so we don't have this benefit. In other words, for every point of the discretized sample path, we must sample from probability distributions which takes up more time. This combined with the fact that for $\mu \geq 1$ we have no guarantee that the mean of the weighted estimator is finite makes the weighted method a poor choice for pricing options. We have a much better choice in the exact sampling method.

%\section{Conclusions} \label{sect6}

%%%%%%%%%%%%%%%%%%%%%%%%%%%%%%%%%%%%%%%%%%%%%%%%%%%%%%%%%%%%%%%%%%%%%%

\end{document}